\documentclass[envcountsame,envcountsect,runningheads]{llncs}

\usepackage{graphicx}
\usepackage{multirow}
\usepackage{latexsym}
\usepackage{amsmath,amssymb,amstext}
\usepackage{graphicx}
\usepackage{wasysym}
\usepackage{float}
\usepackage{enumerate}
\usepackage[arrow, matrix, curve]{xy}
\usepackage[justification=centering,singlelinecheck=off]{caption}
\usepackage{color}
\usepackage{url}
\usepackage{longtable}
\usepackage{wrapfig}
\usepackage{sidecap}

\input{macros.def}

\begin{document} 

\title{Complexity of Fractran and Productivity}
   
\author{J\"{o}rg Endrullis\inst{1} \and Clemens Grabmayer\inst{2} \and Dimitri Hendriks\inst{1}}
\institute{
    Vrije Universiteit Amsterdam,
    Department of Computer Science\\
    De Boelelaan 1081a,
    1081 HV Amsterdam,
    The Netherlands\\
    \email{joerg@few.vu.nl}
      \mbox{ }
    \email{diem@cs.vu.nl}
  \and
    Universiteit Utrecht,
    Department of Philosophy\\
    Heidelberglaan 8,
    3584 CS Utrecht,
    The Netherlands\\
    \email{clemens@phil.uu.nl}
}
\maketitle

\begin{abstract}
  In functional programming languages 
  the use of infinite structures is common practice.
  For total correctness of programs dealing with infinite structures
  one must guarantee that every finite part of the result can be 
  evaluated in finitely many steps.
  This is known as productivity.
  For programming with infinite structures, 
  productivity is what termination in well-defined results
  is for programming with finite structures. 

  Fractran is a simple Turing-complete programming language
  invented by Conway. 
  We prove that the question whether 
  a Fractran program halts on all positive integers
  is \mbox{$\cpi{0}{2}$-complete}.
  In functional programming, productivity typically is a property 
  of individual terms with respect to the inbuilt evaluation strategy.
  By encoding Fractran programs as specifications of infinite lists,
  we establish that this notion of productivity is $\cpi{0}{2}$-complete
  even for some of the most simple specifications.
  Therefore it is harder than termination of individual terms.
  In addition, we explore generalisations of the notion of productivity,
  and prove that their computational complexity is in the analytical hierarchy,
  thus exceeding the expressive power of first-order logic.
\end{abstract}

\section{Introduction}\label{sec:intro}
For programming with infinite structures, productivity is 
what termination is for programming with finite structures.
In lazy functional programming languages like Haskell, Miranda or Clean
the use of data structures, whose intended semantics is an
infinite structure, is common practice.
Programs dealing with such infinite structures
can very well be terminating. 
For example, consider the Haskell program implementing a version of Eratosthenes' sieve: 
\begin{verbatim}
   prime n = primes !! (n-1)
   primes = sieve [2..]
   sieve (n:xs) = n:(sieve (filter (\m -> m `mod` n /= 0) xs))
\end{verbatim}
where \verb=prime n= returns the $n$-th prime number for every $n \ge 1$.
The function \verb=prime= is terminating,
despite the fact that it contains a call to the non-terminating function \verb=primes= 
which, in the limit, rewrites to the infinite list of prime numbers in ascending order.
To make this possible, the strategy with respect to which the terms are evaluated is crucial.
Obviously, we cannot fully evaluate \verb=primes= before extracting the $n$-th element.
For this reason, lazy functional languages
typically use a form of outermost-needed rewriting 
where only needed, finite parts 
of the infinite structure are evaluated, 
see for example~\cite{peyt:1987}.

Productivity captures the intuitive notion of unlimited progress, 
of `working' programs producing values indefinitely,
programs immune to livelock and deadlock, like \verb=primes= above.
A recursive specification is called productive if
not only can the specification be evaluated continually to build up
an infinite normal form, but this infinite expression is also
meaningful in the sense that it represents an infinite object
from the intended domain.
The study of productivity (of stream specifications in particular)
was pioneered by Sijtsma~\cite{sijt:1989}.
More recently, a decision algorithm for productivity of stream specifications 
from an expressive syntactic format has been developed~\cite{endr:grab:hend:isih:klop:2007} 
and implemented~\cite{propro}.

We consider various variants of the notion of productivity and 
pinpoint their computational complexity in the arithmetical and analytical hierarchy.
In functional programming, expressions are evaluated according to an inbuilt evaluation strategy.
This gives rise to \emph{productivity with respect to an evaluation strategy}.
We show that this property is $\cpi{0}{2}$-complete (for individual terms)
using a standard encoding of Turing machines into term rewriting systems.
Next, we explore two generalisations of this concept: \emph{strong} and \emph{weak productivity}.
Strong productivity requires every outermost-fair rewrite sequence to `end
in' a constructor normal form, 
whereas weak productivity demands only the existence of a rewrite sequence
to a constructor normal form.
As it turns out, these properties are of analytical complexity: $\cpi{1}{1}$ and $\csig{1}{1}$-complete,
respectively.

Finally, we encode Fractran programs into stream specifications.
In contrast to the encoding of Turing machines,
the resulting specifications are of a very simple form
and do not involve any computation on the elements of the stream.
We show that the uniform halting problem of Fractran programs is $\cpi{0}{2}$-complete.
(Although Turing-completeness of Fractran is folklore, the exact complexity
has not yet been investigated before.)
Consequently we obtain a strengthening of the earlier mentioned $\cpi{0}{2}$-completeness result
for productivity.

Fractran~\cite{conw:1987} is a remarkably simple Turing-complete programming language 
invented by the mathematician John Horton Conway.
A Fractran program is a finite list of fractions 
$\iafrac{1},\ldots,\iafrac{k}$. 
Starting with a positive integer $n_0$, the algorithm successively calculates $n_{i+1}$
by multiplying $n_i$ with the first fraction that yields an integer again.
The algorithm halts if there is no such fraction.

To illustrate the algorithm we consider an example of Conway from~\cite{conw:1987}:
\[\frac{17}{91}, \frac{78}{85}, \frac{19}{51}, \frac{23}{38},
  \frac{29}{33}, \frac{77}{29}, \frac{95}{23}, \frac{77}{19},
  \frac{1}{17}, \frac{11}{13}, \frac{13}{11}, \frac{15}{14}, \frac{15}{2}, \frac{55}{1}\]
We start with $n_0 = 2$. 
The leftmost fraction which yields an integer product is $\frac{15}{2}$, 
and so $n_1 = 2\cdot \frac{15}{2} = 15$. 
Then we get $n_2 = 15 \cdot \frac{55}{1} = 825$, etcetera.
By successive application of the algorithm, we obtain the following infinite sequence:
\[
  2, 15, 825, 725, 1925, 2275, 425, 390, 330, 290, 770,\ldots
\]
Apart from $2^1$, the powers of $2$ occurring in this infinite sequence 
are $2^2, 2^3, 2^5, 2^7$, $2^{11}, 2^{13}, 2^{17}, 2^{19},\ldots$,
where the exponents form the sequence of primes.

We translate Fractran programs to stream specifications in such a way that
the specification is productive if and only if the program halts on all $n_0 > 1$.
Let us define the target format of this translation:
the \emph{lazy stream format} (LSF\label{page:LSF}).
LSF consists of stream specifications 
of the form $\strff{M} \to \cxtap{\acxt}{\strff{M}}$ 
where $\acxt$ is a context built solely from: 
one data element $\pebble$,
the stream constructor `$\scons$', 
the functions $\head{\cons{x}{\astr}} \to x$ 
and $\tail{\cons{x}{\astr}} \to \astr$,
unary stream functions $\smodn{n}$, 
and $k$-ary stream functions $\szipn{k}$
with the following defining rules, 
for every $n,k \geq 1$:
\begin{gather}
  \begin{aligned}
    \modn{n}{\astr} &\to \cons{\head{\astr}}{\modn{n}{\tailn{n}{\astr}}} 
    \\
    \zipn{k}{\iastr{1},\iastr{2}\ldots,\iastr{k}}
    &\to \cons{\head{\iastr{1}}}{\zipn{k}{\iastr{2},\ldots,\iastr{k},\tail{\iastr{1}}}}
  \end{aligned}
  \tag{LSF}
  \label{LSF}
\end{gather}
By reducing the uniform halting problem of Fractran programs to productivity of LSF,
we get that productivity for LSF is $\cpi{0}{2}$-complete.

This undecidability result stands in sharp contrast to 
the decidability of productivity for the \emph{pure stream format} (PSF,~\cite{endr:grab:hend:isih:klop:2007}).
Let us elaborate on the difference between these two formats.
Examples of specifications in PSF are:
\begin{align*}
  \strff{J} \to \cons{0}{\cons{1}{\even{\strff{J}}}}
  &&\text{and}&&
  \strff{Z} \to \cons{0}{\zip{\even{\strff{Z}}}{\odd{\strff{Z}}}}\punc,
\end{align*}
including the defining rules for the stream functions involved:
\begin{align*}
  \even{\cons{x}{\astr}} \to \cons{x}{\odd{\astr}}\punc, 
  &&
  \odd{\cons{x}{\astr}} \to \even{\astr}\punc,
  && 
  \zip{\cons{x}{\astr}}{\bstr} \to \cons{x}{\zip{\bstr}{\astr}}\punc,
\end{align*}
where $\szip$ `zips' two streams alternatingly into one,
and $\seven$ ($\sodd$) returns a stream consisting of the elements at its even (odd) positions.
The specification for $\strff{Z}$ produces the stream 
$\cons{0}{\cons{0}{\cons{0}{\ldots}}}$ of zeros,
whereas the infinite normal form of $\strff{J}$ is 
$\cons{0}{\cons{1}{\cons{0}{\cons{0}{\sstrev^{\omega}}}}}$,
which is not a constructor normal form.

Excluded from PSF is the observation function on streams $\head{\cons{x}{\astr}} \to x$. 
This is for a good reason, as we shall see shortly. 
PSF is essentially layered: data terms (terms of sort $\datasort$) 
cannot be built using stream terms (terms of sort $\streamsort$).
As soon as \emph{stream dependent} data functions are admitted, 
the complexity of the productivity problem of such an extended format is increased.
Indeed, as our Fractran translation shows,
productivity of even the most simple stream specifications
is undecidable and $\cpi{0}{2}$-hard.
The problem with stream dependent data functions is that they possibly create `look-ahead': 
the evaluation of the `current' stream element may depend on the evaluation of `future' stream elements.
To see this, consider an example from~\cite{sijt:1989}:
\[
  \strff{S}_n \to \cons{0}{\cons{\strnth{\strff{S}_n}{n}}{\strff{S}_n}}
\]
where for a term $t$ of sort stream and $n\in\nat$, 
we write $\strnth{t}{n}$ as a shorthand for $\head{\tailn{n}{t}}$. 
If we take $n$ to be an even number, then $\strff{S}_n$ is productive, 
whereas it is unproductive for odd $n$. 

A hint for the fact that it is $\cpi{0}{2}$-hard to decide 
whether a lazy specification is productive
already comes from a simple encoding of the Collatz conjecture
(also known as the `$3x{+}1$-problem'~\cite{laga:1985})
into a productivity problem. 
Without proof we state: 
\textit{the Collatz conjecture is true if and if only
the following specification produces the infinite chain 
$\cons{\pebble}{\cons{\pebble}{\cons{\pebble}{\ldots}}}$
of data elements $\pebble$:}
\begin{equation}
  \coll \to \cons{\pebble}{\zipn{2}{\coll,\modn{6}{\tailn{9}{\coll}}}}
  \label{eq:collatz}
\end{equation}
In order to understand the operational difference between rules in PSF and rules in LSF,
consider the following two rules:
\begin{align}
  \readlazy{\astr}
  &\to \strcns{\strhd{\astr}}{\readlazy{\strtl{\astr}}}
  \label{eq:read:lazy}
  \\
  \readpure{\strcns{x}{\astr}}
  &\to \strcns{x}{\readpure{\astr}}
  \label{eq:read:pure}
\end{align}
The functions defined by these rules are extensionally equivalent:
they both implement the identity function on fully developed streams. 
However, intensionally, or operationally, there is a difference.
A term $\readpure{\astrtrm}$ is a redex 
only in case $\astrtrm$ is of the form $\strcns{\adattrm}{\bstrtrm}$,
whereas $\readlazy{\astrtrm}$ constitutes a redex 
for \emph{every} stream term $\astrtrm$,
and so $\strhd{\astrtrm}$ can be undefined.
The `lazy' rule~\eqref{eq:read:lazy} \emph{postpones} pattern matching.
Although in PSF we can define functions 
$\spuremodn{n}$ and $\spurezipn{k}$ 
extensionally equivalent to $\smodn{n}$ and $\szipn{k}$,
a pure version $\mkpure{\coll}$ of $\coll$ in~\eqref{eq:collatz} above
(using $\spuremodn{6}$ and $\spurezipn{2}$ instead)
can easily be seen to be not productive (it produces two data elements only),
and to have no bearing on the Collatz conjecture.

\paragraph{Contribution and Overview.}
In Section~\ref{sec:fractran} we show that the uniform halting problem of Fractran programs 
is $\cpi{0}{2}$-complete.
This is the problem of determining whether a program terminates on all positive integers.
Turing-completeness of a computational model does not imply that 
the uniform halting problem
in the strong sense of termination on \emph{all configurations} is $\cpi{0}{2}$-complete.
For example, assume that we extend Turing machines with a special non-terminating state.
Then the computational model obtained can still compute 
every recursive function.
However, the uniform halting problem becomes trivial.

Our result is a strengthening of the result in~\cite{kurt:simo:2007} 
where it has been shown that the generalised Collatz problem (GCP) is $\cpi{0}{2}$-complete.
This is because every Fractran program $\aprg$ can easily be translated into a Collatz function $f$
such that the uniform halting problem for $\aprg$ is equivalent to the GCP for $f$.
The other direction is not immediate, since Fractran programs form a strict subset of Collatz functions.
We discuss this in more detail in Section~\ref{sec:fractran}. 

In Section~\ref{sec:productivity} we explore alternative definitions of productivity
and make them precise in the framework of term rewriting.
These can be highly undecidable: `strong productivity' 
turns out to be \mbox{$\cpi{1}{1}$-complete}
and `weak productivity' is \mbox{$\csig{1}{1}$-complete}.
Productivity of individual terms with respect to a computable strategy,
which is the notion used in functional programming,
is \mbox{$\cpi{0}{2}$-complete}.

In Section~\ref{sec:progs2specs} we prove that productivity $\cpi{0}{2}$-complete even for 
specifications of the restricted LSF format.
The new proof uses a simple encoding of Fractran programs $\aprg$
into stream specifications of the form $\transl{\aprg}\to\acxtap{\transl{\aprg}}$,
in such a way that $\transl{\aprg}$ is productive 
if and only if the program $\aprg$ halts on all inputs.
The resulting stream specifications are very simple compared to the 
ones resulting from encoding of Turing machines employed in Section~\ref{sec:productivity}.
Whereas the Turing machine encoding essentially uses calculations
on the elements of the list,
the specifications obtained from the Fractran encoding contain no operations 
on the list elements.
In particular, the domain of data elements is a singleton.

\paragraph{Related Work.}
In~\cite{endr:geuv:zant:2009} undecidability of different properties
of first-order TRSs is analysed.
While the standard properties of TRSs turn out to be
either $\csig{0}{1}$- or $\cpi{0}{2}$-complete,
the complexity of the dependency pair problems~\cite{arts:gies:00} is essentially analytical:
it is shown to be $\cpi{1}{1}$-complete.
We employ the latter result as a basis for our
$\cpi{1}{1}$- and $\csig{1}{1}$-completeness results for productivity.

Ro\c{s}u~\cite{rosu:2006} shows that equality of stream specifications
is $\cpi{0}{2}$-complete. We remark that this result can be obtained
as a corollary of our translation of Fractran programs $\aprg$ to stream specifications 
$\transl{\aprg}$. Stream specifications $\transl{\aprg}$ 
have the stream $\strcns{\bullet}{\strcns{\bullet}{\ldots}}$ as unique solutions
if and only if they are productive.
Thus $\cpi{0}{2}$-completeness of productivity of these specifications
implies $\cpi{0}{2}$-completeness of the stream equality problem 
$\transl{\aprg} = \strcns{\bullet}{\strcns{\bullet}{\ldots}}$.

One of the reviewers pointed us to recent work~\cite{simo:2009} of \mbox{Grue Simonsen}
(not available at the time of writing)
where $\cpi{0}{2}$-completeness of productivity of orthogonal stream specifications is shown.
Theorem~\ref{thm:strategy} below can be seen as a sharpening of that result
in that we consider general TRSs and productivity with respect to arbitrary evaluation strategies.
For orthogonal systems the evaluation strategy is irrelevant as long as it is outermost-fair.
Moreover we further strengthen the result on orthogonal stream specifications by
restricting the format to LSF.

\section{Fractran}\label{sec:fractran}

The one step computation of a Fractran program is a partial function.
\begin{definition}\normalfont\label{def:fractran:function}
  Let $\aprg = \iafrac{1},\ldots,\iafrac{k}$ be a Fractran program.
  The partial function $f_\aprg \funin \nat \pto \nat$ is defined for all $n\in\nat$ by:
  \begin{equation*}
    \funap{f_\aprg}{n} =
    \begin{cases}
      n\cdot\iafrac{i} & \text{where $\iafrac{i}$ is the first fraction of $\aprg$ 
                               such that $n\cdot\iafrac{i}\in\nat$,} \\
      \text{undefined}          & \text{if no such fraction exists.}
    \end{cases}
  \end{equation*}
  We say that $\aprg$ \emph{halts} on $n \in \nat$ if there exists $i \in \nat$ 
  such that $\funap{f_\aprg^{\hspace{.08em}i}}{n} = \text{undefined}$.
  For $n,m \in \nat$ we write $n \to_{\aprg} m$ whenever $m = \funap{f_\aprg}{n}$.
\end{definition}

The Fractran program for generating prime numbers, that we discussed in the introduction,
is non-terminating for all starting values $n_0$,
because the product of any integer with $\frac{55}{1}$ is an integer again.
However, in general, termination of Fractran programs is undecidable.
\begin{theorem}\label{fractran:pi02}
  The uniform halting problem for Fractran programs, that is, deciding 
  whether a program halts for every starting value $n_0 \in \pnat$,
  is $\cpi{0}{2}$-complete. 
\end{theorem}
A related result is obtained in~\cite{kurt:simo:2007} where it is shown
that the generalised Collatz problem (GCP) is $\cpi{0}{2}$-complete, that is,
the problem of deciding for a Collatz function $f$ 
whether for every integer $x > 0$ there exists $i \in \nat$ such that $\funap{f^i}{x} = 1$.
A Collatz function $f$ is a function $f \funin \nat \to \nat$ of the form:
\begin{align*}
  \funap{f}{n} = 
  \begin{cases}
     a_0 \cdot n + b_0, &\text{if $\congrmod{n}{0}{p}$} \\[-.5ex]
     \vdots & \vdots\\[-.5ex]
     a_{p-1} \cdot n + b_{p-1}, &\text{if $\congrmod{n}{p-1}{p}$}
  \end{cases}
\end{align*}
for some $p \in \nat$ and rational numbers $a_i$, $b_i$ 
such that $\funap{f}{n} \in \nat$ for all $n\in \nat$.

The result of~\cite{kurt:simo:2007} is an immediate corollary of Theorem~\ref{fractran:pi02}.
Every Fractran program $\aprg$ is a Collatz function $f'_\aprg$
where $f'_\aprg$ is obtained from $f_\aprg$ (see Definition~\ref{def:fractran:function})
by replacing undefined with $1$.
We obtain the above representation of Collatz functions simply
by choosing for $p$ the least common multiple of the denominators of the fractions of $\aprg$.
We call a Fractran program $\aprg$ \emph{trivially immortal}
if $\aprg$ contains a fraction with denominator 1 (an integer).
Then for all not trivially immortal $\aprg$, $\aprg$ halts on all inputs
if and only for all $x>0$ there exists $i\in\nat$ such that $\funap{f_\aprg^{\hspace{.08em}i}}{x} = 1$.
Using our result, this implies that GCP is $\cpi{0}{2}$-hard.

Theorem~\ref{fractran:pi02} is a strengthening of the result in~\cite{kurt:simo:2007} 
since Fractran programs are a strict subset of Collatz functions.
If Fractran programs are represented as Collatz functions directly,
for all $0 \le i < p$ it holds either $b_i = 0$, 
or $a_i = 0$ and $b_i = 1$.
Via such a translation Fractran programs are, e.g.,
not able to implement the famous Collatz function 
$\funap{C}{2n} = n$ and $\funap{C}{2n+1} = 6n+4$
(for all $n\in\nat$),
nor an easy function like
$\funap{f}{2n} = 2n+1$ and $\funap{f}{2n+1} = 2n$ 
(for all $n \in \nat$).

For the proof of Theorem~\ref{fractran:pi02} we devise a translation 
from Turing machines to Fractran programs
(\cite{kurt:simo:2007} uses register machines)
such that the resulting Fractran program halts on all positive integers ($n_0 \ge 1$)
if and only if the Turing machine is terminating on all configurations.
That is, we reduce the uniform halting problem of Turing machines
to the uniform halting problem of Fractran programs.

We briefly explain why we employ the uniform halting problem instead of 
the problem of totality (termination on all inputs) of Turing machines, also known as the initialised uniform halting problem.
When translating a Turing machine $\tm$ to a Fractran program $\fprog_\tm$,
start configurations (initialised configurations) are mapped to a subset $I_\tm \subseteq \nat$ of Fractran inputs.
Then from $\cpi{0}{2}$-hardness of the totality problem
one can conclude $\cpi{0}{2}$-hardness
of the question whether $\fprog_\tm$ terminates on all numbers from $I_\tm$.
But this does not imply that the uniform halting problem for Fractran programs is $\cpi{0}{2}$-hard
(termination on all natural numbers $n \in \nat$).
The numbers not in the target of the translation
could make the problem both harder as well as easier.
A situation where extending the domain of inputs makes the
problem easier is: local confluence of TRSs 
is $\cpi{0}{2}$-complete for the set of ground terms,
but only $\csig{0}{1}$-complete for the set of all terms~\cite{endr:geuv:zant:2009}.

To keep the translation as simple we restrict to unary Turing machines 
having only two symbols $\{0,1\}$ in their tape alphabet,
$0$ being the blank symbol.
\begin{definition}\normalfont\label{def:tm}
  A unary \emph{Turing machine $\tm$} 
  is a triple $\triple{\tmstates}{\tmstart}{\stmtrans}$,
  where $\tmstates$ is a finite set of states, $q_0 \in \tmstates$ the initial state,
  and $\stmtrans \funin \tmstates \times \tmsig \pto \tmstates \times \tmsig \times \{\tmL,\tmR\}$
  a (partial) \emph{transition function}.
  A \emph{configuration} of $\tm$ is a pair $\pair{\astate}{\stmtape}$
  consisting of a state $\astate \in \tmstates$
  and the tape content $\stmtape \funin \zz \to \tmsig$
  such that the support $\{n \in \zz \where \tmtape{n} \ne \stmblank\}$ is finite.
  The set of all configurations is denoted by $\tmconf{\tm}$.
  We define the relation $\tmstep$ on the set of configurations $\tmconf{\tm}$ as follows:
  $\pair{\astate}{\stmtape} \tmstep \pair{\astate'}{\stmtape'}$
  whenever:
  \begin{itemize}
   \item
     $\tmtrans{\astate}{\tmtape{0}} = \triple{\astate'}{f}{\tmL}$,
     $\funap{\stmtape'}{1} = f$ and $\myall{n \ne 0}{\funap{\stmtape'}{n+1} = \tmtape{n}}$, or
   \item
     $\tmtrans{\astate}{\tmtape{0}} = \triple{\astate'}{f}{\tmR}$,
     $\funap{\stmtape'}{-1} = f$ and $\myall{n \ne 0}{\funap{\stmtape'}{n-1} = \tmtape{n}}$.
  \end{itemize}
  We say that $\tm$ \emph{halts (or terminates) on a configuration $\pair{\astate}{\stmtape}$}
  if the configuration $\pair{\astate}{\stmtape}$ does not admit infinite $\tmstep$ rewrite sequences.
\end{definition}
The \emph{uniform halting problem} of Turing machines 
is the problem of deciding whether a given Turing machine $\tm$ halts 
on all (initial or intermediate) configurations.
The following theorem is a result of~\cite{herman:1971}:
\begin{theorem}
\label{thm:tm}
  The uniform halting problem for Turing machines is $\cpi{0}{2}$-complete.
\end{theorem}
This result carries over to unary Turing machines using a simulation based on a straightforward encoding
of tape symbols as blocks of zeros and ones (of equal length),
which are admissible configurations of unary Turing machines.

We now give a translation of Turing machines to Fractran programs.
Without loss of generality we restrict in the sequel to Turing machines 
$\tm = \triple{\tmstates}{\tmstart}{\stmtrans}$
for which $\tmtrans{\astate}{x} = \triple{\astate'}{s'}{d'}$ implies $\astate \ne \astate'$.
In case $\tm$ does not fulfil this condition then we can find an equivalent Turing machine 
$\tm' = \triple{\tmstates \cup \tmstates_\#}{\tmstart}{\stmtrans'}$ 
where $\tmstates_\# = \{\astate_\# \where \astate \in \tmstates\}$
and $\stmtrans'$ is defined by 
$\bfunap{\stmtrans'}{\astate}{x} = \triple{p_\#}{s}{d}$ and 
$\bfunap{\stmtrans'}{\astate_\#}{x} = \triple{p}{s}{d}$ 
for $\bfunap{\stmtrans}{\astate}{x} = \triple{p}{s}{d}$.
\begin{definition}\normalfont\label{tm2fractran}
  Let $\tm = \triple{\tmstates}{\tmstart}{\stmtrans}$ be a Turing machine.
  Let $\tapel$, $\mhead$, $\taper$, $\tapel'$, $\mhead'$, $\taper'$, $\movel{x}$, $\mover{x}$, $\mcopy{x}$
  and $p_q$ for every $q \in \tmstates$ and $x \in \{0,1\}$ be pairwise distinct prime numbers.
  The intuition behind these primes is:
  \begin{itemize}
    \item $\tapel$ and $\taper$ represent the tape left and right of the head, respectively,
    \item 	$\mhead$ is the tape symbol in the cell currently scanned by the tape head,
    \item $\tapel'$, $\mhead'$, $\taper'$ store temporary tape content (when moving the head),
    \item $\movel{x}$, $\mover{x}$ execute a left or right move of the head on the tape, respectively,
    \item $\mcopy{x}$ copies the temporary tape content back to the primary tape, and
    \item $p_q$ represent the states of the Turing machine.
  \end{itemize}
  The subscript $x \in \{0,1\}$ is used to have two primes for every action:
  in case an action $p$ takes more than one calculation step
  we cannot write $\frac{p\cdot \ldots }{p\cdot \ldots}$ 
  since then $p$ in numerator and denominator
  would cancel itself out.
  We define the Fractran program $\fractm$ to consist of the following fractions 
  (listed in program order):
  \begin{gather}
    \frac{1}{p \cdot p'} \quad\quad
    \begin{aligned}
    \text{for every } & p,p' \in \{\movel{0},\,\movel{1},\,\mover{0},\,\mover{1},\,\mcopy{0},\,\mcopy{1}\}\\[-.3ex]
    \text{every } & p,p' \in \{p_q \where q \in \tmstates\}
    \text{ and } p,p' \in \{\mhead,\mhead'\}
    \end{aligned}
    \label{frac:illegal}
  \end{gather}
  to get rid of illegal configurations,
  \begin{align}
    \frac{\movel{1-x} \cdot \tapel'}{\movel{x} \cdot \tapel^2} &&
    \frac{\movel{1-x} \cdot \taper'^2}{\movel{x} \cdot \taper} &&
    \frac{\movel{1-x} \cdot \taper'}{\movel{x} \cdot \mhead'} &&
    \frac{\movel{1-x} \cdot \mhead}{\movel{x} \cdot \tapel} &&
    \frac{\mcopy{0}}{\movel{x}}
    \label{frac:ml} 
  \end{align}
  with $x \in \{0,1\}$, for moving the head left on the tape,
  \begin{align}
    \frac{\mover{1-x} \cdot \taper'}{\mover{x} \cdot \taper^2} &&
    \frac{\mover{1-x} \cdot \tapel'^2}{\mover{x} \cdot \tapel} &&
    \frac{\mover{1-x} \cdot \tapel'}{\mover{x} \cdot \mhead'} &&
    \frac{\mover{1-x} \cdot \mhead}{\mover{x} \cdot \taper} &&
    \frac{\mcopy{0}}{\mover{x}}
    \label{frac:mr}
  \end{align}
  with $x \in \{0,1\}$, for moving the head right on the tape,
  \begin{align}
    \frac{\mcopy{1-x} \cdot \tapel}{\mcopy{x} \cdot \tapel'} &&
    \frac{\mcopy{1-x} \cdot \taper}{\mcopy{x} \cdot \taper'} &&
    \frac{1}{\mcopy{x}}
    \label{frac:copy}
  \end{align}
  with $x \in \{0,1\}$, for copying the temporary tape back to the primary tape,
  \begin{align}
    \frac{p_{q'} \cdot \mhead'^{s'} \cdot m_{d,0}}{p_{q} \cdot \mhead} 
    &&& \text{whenever $\tmtrans{\astate}{1} = \triple{\astate'}{s'}{d}$ } \label{frac:tmh}\\
    \frac{1}{p_{q} \cdot \mhead} \hspace{5.5ex}
    &&& \text{(for termination) for every $\astate \in \tmstates$ } \label{frac:term}\\
    \frac{p_{q'} \cdot \mhead'^{s'} \cdot m_{d,0}}{p_{q}} 
    &&& \text{whenever $\tmtrans{\astate}{0} = \triple{\astate'}{s'}{d}$} \label{frac:tmnoh}
  \end{align}
  for the transitions of the Turing machine.
  Whenever we use variables in the rules, e.g. $x \in \{0,1\}$, then it is to be understood
  that instances of the same rule are immediate successors in the sequence of fractions
  (the order of the instances among each other is not crucial).
\end{definition}

\begin{example}
  Let $\tm = \triple{\tmstates}{a_0}{\stmtrans}$ be a Turing machine 
  where $\tmstates = \{a_0,a_1,b\}$, and
  the transition function is defined by
  $\tmtrans{a_0}{0} = \triple{b}{1}{\tmR}$,
  $\tmtrans{a_1}{0} = \triple{b}{1}{\tmR}$,
  $\tmtrans{a_0}{1} = \triple{a_1}{0}{\tmR}$,
  $\tmtrans{a_1}{1} = \triple{a_0}{0}{\tmR}$,
  $\tmtrans{b}{1} = \triple{a_0}{0}{\tmR}$,
  and we leave $\tmtrans{b}{0}$ undefined.
  That is, $\tm$ moves to the right,
  converting zeros into ones and vice versa, 
  until it finds two consecutive zeros and terminates.
  Assume that $\tm$
  is started on the configuration $1b1001$,
  that is, the tape content $11001$ in state $b$ with the head located on the second $1$.
  In the Fractran program $\fractm$
  this corresponds to
  $n_0 = p_{b} \cdot \tapel^1 \cdot \mhead^1 \cdot \taper^{100}$
  as the start value
  where we represent the exponents in binary notation for better readability.
  Started on $n_0$ we obtain the following calculation in $\fractm$:
  \begin{gather*}
    \ored{p_{b}} \cdot \tapel^1 \cdot \ored{\mhead^{1}} \cdot \taper^{100} \quad\quad \text{(configuration $1b1001$)}\\
    \to_{\eqref{frac:tmh}} \ored{\mover{0}} \cdot p_{a_0} \cdot \tapel^1 \cdot \ored{\taper^{100}}
    \to_{(\ref{frac:mr};1^\text{st})}^2 \ored{\mover{0}} \cdot p_{a_0} \cdot \ored{\tapel^1} \cdot \taper'^{10}\\
    \to_{(\ref{frac:mr};2^\text{nd})} \ored{\mover{1}} \cdot p_{a_0} \cdot \tapel'^{10} \cdot \taper'^{10}
    \to_{(\ref{frac:mr};5^\text{th})} \ored{\mcopy{0}} \cdot p_{a_0} \cdot \ored{\tapel'^{10}} \cdot \ored{\taper'^{10}}\\
    \to_{(\ref{frac:copy};1^\text{st})}^2 \to_{(\ref{frac:copy};2^\text{nd})}^2 \to_{(\ref{frac:copy};3^\text{rd})}
      \ored{p_{a_0}} \cdot \tapel^{10} \cdot \taper^{10} \quad\quad \text{(configuration $10a001$)}\\
    \to_{\eqref{frac:tmnoh}} \ored{\mover{0}} \cdot p_{b} \cdot \tapel^{10} \cdot \mhead'^{1} \cdot \ored{\taper^{10}}
    \to_{(\ref{frac:mr};1^\text{st})} \ored{\mover{1}} \cdot p_{b} \cdot \ored{\tapel^{10}} \cdot \mhead'^{1} \cdot \taper'^{1}\\
    \to_{(\ref{frac:mr};2^\text{nd})}^2 \ored{\mover{1}} \cdot p_{b} \cdot \tapel'^{100} \cdot \ored{\mhead'^{1}} \cdot \taper'^{1}
    \to_{(\ref{frac:mr};3^\text{rd} +5^\text{th})}
     \ored{\mcopy{0}} \cdot p_{b} \cdot \ored{\tapel'^{101}} \cdot \ored{\taper'^{1}}\\
    \to_{(\ref{frac:copy};1^\text{st})}^5 \to_{(\ref{frac:copy};2^\text{nd})} \to_{(\ref{frac:copy};3^\text{rd})}
      p_{b} \cdot \tapel^{101} \cdot \taper^{1} \quad\quad \text{(configuration $101b01$)}
  \end{gather*}
  reaching a configuration where the Fractran program halts.
\end{example}

\begin{definition}\label{def:fracconf}\normalfont
  We translate configurations $c = \pair{\astate}{\stmtape}$ of Turing machines $\tm = \triple{\tmstates}{\tmstart}{\stmtrans}$
  to natural numbers (input values for Fractran programs).
  We reuse the notation of Definition~\ref{tm2fractran} and define:
  \begin{gather*}
  n_c = \tapel^L \cdot p_{\astate} \cdot \mhead^H \cdot \taper^R\\
  L = \sum_{i = 0}^\infty 2^i \cdot \tmtape{-1 - i} \quad\quad
  H = \tmtape{0} \quad\quad
  R = \sum_{i = 0}^\infty 2^i \cdot \tmtape{1+i}
  \end{gather*}
\end{definition}

\begin{lemma}\label{lem:tm2fractran:step}
   For every Turing machine $\tm$ and
   configurations $c_1$, $c_2$
   we have:
   \begin{enumerate}
      \item \label{tm2fractran:i} if $c_1 \tmstep c_2$ then $n_{c_1} \to_{\fractm}^* n_{c_2}$, and
      \item \label{tm2fractran:ii} if $c_1$ is a $\tmstep$ normal form then $n_{c_1} \to_{\fractm}^* \text{undefined}$.
   \end{enumerate}
\end{lemma}
\begin{proof}
  Let $\tm = \triple{\tmstates}{\tmstart}{\stmtrans}$, and $c_i = \pair{\astate_i}{\stmtape_i}$ for $i \in \{1,2\}$.
  Then for $i \in \{1,2\}$ we have $n_{c_i} = \tapel^{L_i} \cdot p_{\astate_i} \cdot \mhead^{H_i} \cdot \taper^{R_i}$
  with $L_i$, $H_i$ and $R_i$ as $L$, $H$ and $R$ in Definition~\ref{def:fracconf}.
  For~\ref{tm2fractran:i} assume that $c_1 \tmstep c_2$.
  By the definition of $\tmstep$ there are two cases: the head moves left or right.
  We consider `moving left' (`moving right' is analogous).
  Then
  $\tmtrans{\astate_1}{\funap{\stmtape_1}{0}} = \triple{\astate_2}{s}{\tmL}$,
  $\funap{\stmtape_2}{1} = s$, $\myall{n \ne 0}{\funap{\stmtape_2}{n+1} = \funap{\stmtape_1}{n}}$.
  Therefore $L_2 = \floor{L_1/2}$, $H_2 = \funap{\stmtape_1}{-1} = L_1 \mod 2$ and $R_2 = s + 2\cdot R_1$.
  If $\funap{\stmtape_1}{0} = 1$ then $H_1 = 1$ and therefore the first fraction of $\fractm$ applicable to $n_{c_1}$ is
  $\sfrac{p_{q_2} \cdot \mhead'^{s} \cdot \movel{0}}{p_{q_1} \cdot \mhead}$;
  otherwise if $\funap{\stmtape_1}{0} = 0$ then $H_1 = 0$ and the fraction is 
  $\sfrac{p_{q_2} \cdot \mhead'^{s} \cdot \movel{0}}{p_{q_1}}$.
  Both cases result in (numbers in the justification refer to fractions of Definition~\ref{tm2fractran}): 
  \begin{align*}
    n_{c_1} &\to_{\fractm} p_{q_2} \cdot \tapel^{L_1} \cdot \mhead'^{s} \cdot \taper^{R_1} \cdot \movel{0}
            &&\text{(\eqref{frac:tmh} or \eqref{frac:tmnoh})}\\
            &\to_{\fractm}^* p_{q_2} \cdot \tapel^{L_1 \mod 2} \cdot \tapel'^{\floor{L_1/2}} \cdot \mhead'^{s} \cdot \taper^{R_1} \cdot \movel{x}
            &&\text{($1^{\text{st}}$ of \eqref{frac:ml})} \\
            &\to_{\fractm}^* p_{q_2} \cdot \tapel^{L_1 \mod 2} \cdot \tapel'^{\floor{L_1/2}} \cdot \mhead'^{s} \cdot \taper'^{2\cdot R_1} \cdot \movel{y}
            &&\text{($2^{\text{nd}}$ of \eqref{frac:ml})}\\
            &\to_{\fractm}^{\le 1} p_{q_2} \cdot \tapel^{L_1 \mod 2} \cdot \tapel'^{\floor{L_1/2}} \cdot \taper'^{s + 2\cdot R_1} \cdot \movel{y}
            &&\text{($3^{\text{rd}}$ of \eqref{frac:ml})}\\
            &\to_{\fractm}^{\le 1} p_{q_2} \cdot \tapel'^{\floor{L_1/2}} \cdot \mhead^{L_1 \mod 2} \cdot \taper'^{s + 2\cdot R_1} \cdot \movel{z}
            &&\text{($4^{\text{th}}$ of \eqref{frac:ml})}\\
            &\to_{\fractm}^{\le 1} p_{q_2} \cdot \tapel'^{\floor{L_1/2}} \cdot \mhead^{L_1 \mod 2} \cdot \taper'^{s + 2\cdot R_1} \cdot \mcopy{0}
            &&\text{($5^{\text{th}}$ of \eqref{frac:ml})}\\
            &\to_{\fractm}^{*} p_{q_2} \cdot \tapel^{\floor{L_1/2}} \cdot \mhead^{L_1 \mod 2} \cdot \taper^{s + 2\cdot R_1} = n_{c_2}
            &&\text{(\eqref{frac:copy})}
  \end{align*}
  For~\ref{tm2fractran:ii} assume that $c_1$ is a $\tmstep$ normal form.
  Then $\tmtrans{\astate_1}{\funap{\stmtape_1}{0}}$ is undefined.
  If $\funap{\stmtape_1}{0} = 1$ then $H_1 = 1$. 
  Since there is no matching fraction~\eqref{frac:tmh},
  the first applicable fraction is from~\eqref{frac:term}, which removes $p_{\astate_1}$ and thus leads to termination.
  If $\funap{\stmtape_1}{0} = 0$ then $H_1 = 0$ thus the only applicable fractions can be among~\eqref{frac:tmnoh}
  however there is no matching fraction since $\tmtrans{\astate_1}{\funap{\stmtape_1}{0}}$ is undefined.
  \qed
\end{proof}

Finally, we prove Theorem~\ref{fractran:pi02} which states that 
the uniform halting problem for Fractran programs is $\cpi{0}{2}$-complete.
\begin{proof}[Theorem~\ref{fractran:pi02}]
  For $\cpi{0}{2}$-hardness use the uniform halting problem of Turing machines which
  is $\cpi{0}{2}$-complete (see Theorem~\ref{thm:tm}).
  Let $\tm$ be a Turing machine.
  We prove that $\tm$ halts on all configurations if and only if $\fractm$ halts on all integers $n > 0$.
  If there is a configuration $c$ on which $\tm$ does not halt, then
  the Fractran program $\fractm$ does not halt on $n_c$ by Lemma~\ref{lem:tm2fractran:step} \ref{tm2fractran:i}.
  Thus assume that $\tm$ halts on all configurations.
  Let $C = \{\movel{0},\,\movel{1},\,\mover{0},\,\mover{1},\,\mcopy{0},\,\mcopy{1}\}$.
  Let $n > 0$ be arbitrary.
  By Lemma~\ref{lem:tm2fractran:step} it suffices to show
  that $n \to_{\fractm}^* \text{undefined}$ or $n \to_{\fractm}^* n_c$
  for some configuration $c$.
  By the first fractions of $\fractm$ we have $n \to_{\fractm}^* n'$
  such that $n'$ contains at most one prime factor from $C$, at most one $p_\astate$ and at most one $\{\mhead,\mhead'\}$
  (and none of these primes has an exponent $> 1$).

  Assume that $n'$ contains $\smovel$ or $\smovel$.
  The $\smovel$ or $\smovel$ fractions cannot be applied infinitely often in sequence
  since they decrease all exponents of $\tapel'$, $\taper'$ and $\mhead'$ to $0$, respectively.
  After $\smovel$ or $\smover$ there always follows $\mcopy{}$
  which then increases $\tapel$, $\taper$, but decreases $\tapel'$, $\taper'$ to the value $0$
  and afterwards $\mcopy{}$ `removes' itself. We call the reached configuration $n''$.
  Then $n''$ contains only the prime factors
  $\tapel$, $\taper$ with exponent $\ge 0$, $\mhead$ with exponent $\le 1$ and at most one of the $p_\astate$ with exponent $\le 1$.
  If $n''$ does not contain any $p_q$ then $n'' \to_{\fractm} \text{undefined}$.
  Otherwise there exists a configuration $c$ such that $n'' = n_c$.

  If $n'$ does not contain $\smovel$ or $\smover$, then
  neither the fractions~\eqref{frac:illegal}, nor~\eqref{frac:copy} ($\mcopy{}$),
  can be applied infinitely often in sequence. 
  Application of~\eqref{frac:term} removes
  the only remaining $p_{\astate}$ and thus leads to termination.
  Thus any non-terminating sequence contains an application of~\eqref{frac:tmh} or~\eqref{frac:tmnoh},
  which brings us back to the case of $n'$ containing $\smovel$ or $\smover$ which we have already analysed.
  \qed
\end{proof}

\section{What is Productivity?}\label{sec:productivity}
A program is productive if it evaluates to a finite or infinite constructor normal form.
This rather vague description leaves open several choices 
that can be made to obtain a more formal definition. 
We explore several definitions and determine the degree of undecidability for each of them.
See~\cite{endr:grab:hend:isih:klop:2007}
for more pointers to the literature on productivity.

The following is a productive specification of the (infinite) stream of zeros:
\begin{align*}
  \zeros &\to \strcns{\zer}{\zeros}
\end{align*}
Indeed, there exists only one maximal 
rewrite sequence from $\zeros$ and this ends in the infinite constructor normal form 
$\strcns{0}{\strcns{0}{\strcns{0}{\ldots}}}$.
Here and later we say that a rewrite sequence $\rho : t_0 \to t_1 \to t_2 \to \ldots$ 
\emph{ends in} a term $s$ if either $\rho$ is finite with its last term being $s$,
or $\rho$ is infinite and then $s$ is the limit of the sequence of terms~$t_i$, i.e.\
$s = \lim_{i\to\infty}t_i$.
We consider only rewrite sequences starting from finite terms, thus all
terms occurring in $\rho$ are finite. Nevertheless, the limit $s$ of the terms $t_i$
may be an infinite term.
Note that, if $\rho$ ends in a constructor normal form, 
then every finite prefix will be evaluated after finitely many steps.

The following is a slightly modified specification of the stream of zeros:
\begin{align*}
  \zeros &\to \strcns{\zer}{\funap{\strff{id}}{\zeros}} & \funap{\strff{id}}{\astr} &\to \astr
\end{align*}
This specification is considered productive as well,
although there are infinite rewrite sequences that do not even end in a normal form,
let alone in a constructor normal form: e.g.\ by 
unfolding $\zeros$ only we get the limit term 
$\strcns{0}{\funap{\strff{id}}{\strcns{0}{\funap{\strff{id}}{\strcns{0}{\funap{\strff{id}}{\ldots}}}}}}$.
In general, normal forms can only be reached by outermost-fair rewriting sequences.
A rewrite sequence $\rho : t_0 \to t_1 \to t_2 \to \ldots$ is \emph{outermost-fair}~\cite{terese:2003}
if there is no $t_n$ containing an outermost redex which remains an outermost redex infinitely long,
and which is never contracted.
For this reason it is natural to consider productivity of terms with respect to outermost-fair strategies.

What about stream specifications that admit rewrite sequences to constructor normal forms, 
but that also have divergent rewrite sequences:
\begin{align*}
  \maybe &\to \strcns{0}{\maybe} &
  \maybe &\to \sink &
  \sink &\to \sink
\end{align*}
This example illustrates that, for non-orthogonal stream specifications,
reachability of a constructor normal form depends on the evaluation strategy.
The term $\maybe$ is only productive
with respect to strategies that always apply the first rule.

For this reason we propose to think of productivity as
a property of individual terms 
\emph{with respect to} a given rewrite strategy.
This reflects the situation in functional programming,
where expressions are evaluated according to an inbuilt strategy.
These strategies are usually based on a form of outermost-needed rewriting 
with a priority order on the rules.

\subsection{Productivity with respect to Strategies}
  \label{sec:productivity:subsec:strategies}

For term rewriting systems (TRSs) \cite{terese:2003}
we now fix definitions of the notions of 
(history-free) strategy and history-aware strategy.
Examples for the latter notion are outermost-fair strategies,
which typically have to take history into account.

\begin{definition}\normalfont
  Let $\atrs$ be a TRS with rewrite relation $\redR{\atrs}$.

  A \emph{strategy for $\redR{\atrs}$} is a relation 
  $\sastrat\subseteq\sredR{\atrs}$ with the same normal forms
  as $\redR{\atrs}$.

  The \emph{history-aware rewrite relation $\sredh{\atrs}$ for $\atrs$}
  is the binary relation on $\ter{\asig} \times (\atrs \times \nat^*)^*$ 
  that is defined by:
  \begin{align*}
    \pair{s}{\ahist_s} \redh{\atrs} \pair{t}{\lstcns{\ahist_s}{\pair{\rho}{\apos}}} &\Longleftrightarrow
    s \red t \text{ via rule $\rho \in \atrs$ at position $\apos$}
    \punc.
  \end{align*}
  We identify $t \in \ter{\asig}$ with $\pair{t}{\posemp}$,
  and for $s,t \in \ter{\asig}$ we write $s \redh{\atrs} t$ 
  whenever $\pair{s}{\posemp} \redh{\atrs} \pair{t}{\ahist}$ for some history $h \in (\atrs \times \nat^*)^*$.
  A \emph{history-aware strategy for $\atrs$} is a strategy for
  $\sredh{\atrs}$.

  A strategy $\astrat$ is \emph{deterministic} if 
  $s \astrat t$ and $s \astrat t'$ implies $t = t'$. 
  A strategy $\astrat$ is \emph{computable} if the function mapping
  a term (a term/history pair) to its set of $\astrat$-successors
  is a total recursive function, after coding into natural numbers. 
\end{definition}

\begin{remark}
  Our definition of strategy for a rewrite relation
  follows \cite{toya:1992}.
  For abstract rewriting systems, 
  in which rewrite steps are first-class citizens,
  a definition of strategy is given in \cite[Ch.\,9]{terese:2003}.
  There, history-aware strategies for a TRS~$\atrs$ are defined in terms of 
  `labellings' for the `abstract rewriting system' underlying $\atrs$.
  While that approach is conceptually advantageous, our definition
  of history-aware strategy is equally expressive.  
\end{remark}

\begin{definition}\normalfont
  A \emph{(TRS-indexed) family of strategies $\stratfam$} 
  is a function that assigns to every TRS $\atrs$ 
  a set $\funap{\stratfam}{\atrs}$  of strategies for $\atrs$.
  We call such a family $\stratfam$ of strategies \emph{admissible} 
  if $\funap{\stratfam}{\atrs}$ 
  is non-empty for every orthogonal TRS $\atrs$.
\end{definition}

\newcommand{\strategy}{\leadsto}

Now we give the definition of productivity with respect to a strategy.
\begin{definition}\normalfont\label{def:productivity}
  A term $t$ is called \emph{productive with respect to a strategy $\strategy$}
  if all maximal $\strategy$ rewrite sequences starting from $t$ end in a constructor normal form.
\end{definition}
In the case of non-deterministic strategies we require here that all maximal rewrite sequences end in a constructor normal form.
Another possible choice could be to require only the existence of 
one such rewrite sequence (see Section~\ref{sec:productivity:subsec:strong}).
However, we think that productivity should be a practical notion.
Productivity of a term should entail that arbitrary finite parts of the constructor normal form can indeed be evaluated.
The mere requirement that a constructor normal form exists
leaves open the possibility that such a normal form 
cannot be approximated to every finite precision in a computable way.

For orthogonal TRSs outermost-fair (or fair) rewrite strategies
are the natural choice for investigating productivity
because
they guarantee to find (the unique) infinitary constructor normal form whenever it exists (see~\cite{terese:2003}).

Pairs and finite lists of natural numbers can be encoded using the well-known G\"{o}del encoding.
Likewise terms and finite TRSs over a countable set of variables
can be encoded.
A TRS is called \emph{finite} if its signature and set of rules
are finite.
In the sequel we restrict to (families of) computable strategies,
and assume that strategies are represented by
appropriate encodings.

Now we define the productivity problem in TRSs with respect to families
of computable strategies, and prove a $\cpi{0}{2}$-completeness result.

\paragraph{\textsc{Productivity Problem} \normalfont 
  with respect to a family $\stratfam$ of computable strategies}\ 

\emph{Instance}:
  Encodings of a finite TRS $\atrs$, 
  a strategy ${\strategy} \in \funap{\stratfam}{\atrs}$ 
  and a term $t$.

\emph{Answer}:
  `Yes' if $t$ is productive with respect to $\strategy$. %% , and `No', otherwise.

\begin{theorem}\label{thm:strategy}
  For every family of admissible, computable strategies $\stratfam$,
  the productivity problem with respect to $\stratfam$ 
  is $\cpi{0}{2}$-complete.
\end{theorem}

\begin{proof}{
  \newcommand{\from}{\funap{\msf{from}}}
  \newcommand{\run}{\funap{\msf{M}}}
  A (deterministic) Turing machine is called \emph{total} (encodes a total function $\nat \to \nat$) 
  if it halts on all inputs encoding natural numbers.
  The problem of deciding whether a deterministic Turing machine is \emph{total} is well-known to be $\cpi{0}{2}$-complete, see~\cite{hinman:1978}.
  Let $\tm$ be an arbitrary Turing machine.
  Employing the encoding of deterministic Turing machines into orthogonal TRSs from~\cite{jw:handbook},
  we can define a TRS $\atrs_\tm$ that simulates $\tm$ such that
  for every $n \in \nat$ it holds:
  every reduct of the term $\run{\funap{\ssuc^n}{\zer}}$ contains at most one redex occurrence,
  and the term $\run{\funap{\ssuc^n}{\zer}}$ rewrites to $\zer$ 
  if and only if the Turing machine $\tm$ halts on the input $n$.
  Note that the rewrite sequence starting from $\run{\funap{\ssuc^n}{\zer}}$ is deterministic.
  We extend the TRS $\atrs_\tm$ to a TRS $\atrs'_\tm$ with the following rule:
  \begin{align*}
    \funap{\strff{go}}{\zer,x} &\to \strcns{\zer}{\funap{\strff{go}}{\run{x},\suc{x}}}
  \end{align*}
  and choose the term $t = \funap{\strff{go}}{\zer,\zer}$.
  Then $\atrs'_\tm$ is orthogonal and 
  by construction every reduct of $t$ contains at most one redex occurrence
  (consequently all strategies for $\atrs$ coincide on every reduct of $t$).
  The term $t$ is productive
  if and only if
  $\run{\funap{\ssuc^n}{\zer}}$ rewrites to $\zer$ for every $n \in \nat$
  which in turn holds if and only if
  the Turing machine $\tm$ is total.
  This concludes $\cpi{0}{2}$-hardness.

  For $\cpi{0}{2}$-completeness let $\stratfam$ be a family of computable strategies,
  $\atrs$ a TRS,~${\strategy} \in \funap{\stratfam}{\atrs}$ and $t$ a term. 
  Then productivity of $t$ can be characterised as:
  \begin{equation}
    \tag{$\star$}\label{PI02form:thm:strategy}
    \begin{aligned}
      \myall{d \in \nat}{\myex{n \in \nat}{&\text{every $n$-step $\strategy$-reducts of $t$}}}\\[-0.5ex]
      &\text{is a constructor normal form up to depth $d$}
    \end{aligned}
  \end{equation}
  Since the strategy ${\strategy}$ is computable and finitely branching, 
  all $n$-step reducts of $t$ can be computed.
  Obviously, if the formula \eqref{PI02form:thm:strategy}
  holds, then $t$ is productive w.r.t.\ $\strategy$.
  Conversely, assume that $t$ is productive w.r.t.\ $\strategy$.
  For showing \eqref{PI02form:thm:strategy}, let $d\in\nat$ be arbitrary.
  By productivity of $t$ w.r.t.\ $\strategy$, on every path in the reduction 
  graph of $t$ w.r.t.\ $\strategy$ eventually a term with a constructor normal
  form up to depth $d$ is encountered. Since reduction graphs in TRSs
  always are finitely branching, Koenig's lemma implies that there 
  exists an $n\in\nat$ such that all terms on depth greater or equal to $n$ 
  in the reduction graph of $t$ are constructor prefixes of depth at least $d$. 
  Since $d$ was arbitrary, \eqref{PI02form:thm:strategy} has been established. 
  Because \eqref{PI02form:thm:strategy} is a $\cpi{0}{2}$-formula, 
  the productivity problem with respect to $\stratfam$
  also belongs to $\cpi{0}{2}$.
  \qed
}\end{proof}

\begin{corollary}
  In orthogonal TRSs, productivity is $\cpi{0}{2}$-complete 
  with respect to outermost-fair rewriting. 
\end{corollary}
\begin{proof}
  Apply Theorem~\ref{thm:strategy} to the family of strategies 
  that assigns to every orthogonal TRS $\atrs$
  the set of computable, outermost-fair rewriting strategies for $\atrs$,
  and $\setemp$ to non-orthogonal~TRSs.
\end{proof}

The definition of productivity with respect to computable strategies
reflects the situation in functional programming.
Nevertheless, we now investigate variants of this notion, and determine
their respective computational complexity.

\subsection{Strong Productivity}\label{sec:productivity:subsec:strong}

As already discussed, only outermost-fair rewrite sequences can reach a constructor normal form.
Dropping the fine tuning device `strategies', 
we obtain the following stricter notion of productivity.
\begin{definition}\normalfont\label{def:strongproductivity}
  A term $t$ is called \emph{strongly productive}
  if all maximal outermost-fair rewrite sequences starting from $t$ end in a constructor normal form.
\end{definition}
The definition requires all outermost-fair rewrite sequences to end in a constructor normal form,
including non-computable rewrite sequences.
This catapults 
productivity into a much higher class of undecidability:
$\cpi{1}{1}$, a class of the analytical hierarchy.
The analytical hierarchy continues the classification
of the arithmetical hierarchy using second order formulas.
The computational complexity of strong productivity therefore
exceeds the expressive power of first-order logic 
to define sets from recursive sets.

A well-known result of recursion theory states
that for a given computable relation ${>} \subseteq \nat \times \nat$
it is $\cpi{1}{1}$-hard to decide whether $>$ is well-founded, see~\cite{hinman:1978}.
Our proof is based on a construction from~\cite{endr:geuv:zant:2009}.
There a translation from deterministic Turing machines $\tm$ to TRSs $\roottrs$ (which we explain below)
together with a term $t_\tm$ 
is given such that: $t_\tm$ is root-terminating (i.e., $t_\tm$
admits no rewrite sequences containing an infinite number of root steps)
if and only if the binary relation $>_\tm$ encoded by $\tm$ is well-founded.
The TRS $\roottrs$ consists of rules
for simulating the Turing machine $\tm$ such that $\bfunap{\tm}{x}{y} \red^* \tmT$ iff $x >_\tm y$ holds
(which basically uses a standard encoding of deterministic Turing machines, see~\cite{jw:handbook}), a rule:
\begin{gather*}
  \tfunap{\msf{run}}{\tmT}{\picknok{x}}{\picknok{y}} \to \tfunap{\msf{run}}{\bfunap{\tm}{x}{y}}{\picknok{y}}{\pickn}
\end{gather*}
and rules for randomly generating a natural number:
\begin{align*}
\pickn &\to \fc{\pickn} &\quad
\pickn &\to \picknok{\tmzer{\tmiblank}} &\quad
\fc{\picknok{x}} &\to \picknok{\fs{x}}\punc.
\end{align*}
The term $t_\tm = \tfunap{\msf{run}}{\tmT}{\pickn}{\pickn}$ admits a rewrite sequence containing infinitely
many root steps if and only if $>_\tm$ is not well-founded.
More precisely, whenever there is an infinite decreasing sequence $x_1 >_\tm x_2 >_\tm x_3 >_\tm\ldots$,
then $t_\tm$ admits a rewrite sequence
$\tfunap{\msf{run}}{\tmT}{\pickn}{\pickn} 
 \red^* \tfunap{\msf{run}}{\tmT}{\picknok{x_1}}{\picknok{x_2}}
 \red \tfunap{\msf{run}}{\bfunap{\tm}{x_1}{x_2}}{\picknok{x_2}}{\pickn}
 \red^* \tfunap{\msf{run}}{\tmT}{\picknok{x_2}}{\picknok{x_3}}
 \red^* \ldots$.
We further note that $t_\tm$ and all of its reducts contain
exactly one occurrence of the symbol $\msf{run}$, namely at the root position.

\begin{theorem}\label{thm:omfair}
  The recognition problem for strong productivity is $\cpi{1}{1}$-complete.
\end{theorem}
\begin{proof}
  For the proof of $\cpi{1}{1}$-hardness, let $\tm$ be a deterministic
  Turing machine.
  We extend the TRS $\roottrs$ from~\cite{endr:geuv:zant:2009}
  with the rule $\funap{\msf{run}}{x,y,z} \to \strcns{0}{\funap{\msf{run}}{x,y,z}}$.
  As a consequence the term $\tfunap{\msf{run}}{\tmT}{\pickn}{\pickn}$ is strongly productive
  if and only if ${>_\tm}$ is well-founded (which is $\cpi{1}{1}$-hard to decide).
  If $>_\tm$ is not well-founded, then by the result in~\cite{endr:geuv:zant:2009}
  $t_\tm$ admits a rewrite sequence containing infinitely many root steps which obviously does not end in a constructor normal form.
  On the other hand if $>_\tm$ is well-founded, then $t_\tm$ admits only finitely many root steps with respect to $\roottrs$,
  and thus by outermost-fairness the freshly added rule has to be applied infinitely often.
  This concludes $\cpi{1}{1}$-hardness.

  Rewrite sequences of length $\omega$ can be represented by functions $r \funin \nat \to \nat$
  where $\funap{r}{n}$ represents the $n$-th term of the sequence together with the position and rule
  applied in step $n$. 
  Then for all $r$ (one universal $\forall r$ function quantifier) we have to check
  that $r$ converges towards a constructor normal form whenever $r$ is outermost-fair;
  this can be checked by a first order formula.
  We refer to~\cite{endr:geuv:zant:2009} for the details of the encoding.
  Hence strong productivity is in $\cpi{1}{1}$.
  \qed
\end{proof}

\subsection{Weak Productivity}%
  \label{sec:productivity:subsec:weak}

A natural counterpart to strong productivity is the notion of `weak productivity':
the existence of a rewrite sequence to a constructor normal form.
Here outermost-fairness does not need to be required, 
because rewrite sequences that reach normal forms
are always outermost-fair.

\begin{definition}\normalfont\label{def:weakproductivity}
  A term $t$ is called \emph{weakly productive}
  if there exists a rewrite sequence starting from $t$ that ends in a constructor normal form.
\end{definition}
For non-orthogonal TRSs the practical relevance of this definition is questionable
since, in the absence of a computable strategy to reach normal forms,
mere knowledge that a term $t$ is productive does typically not help
to find a constructor normal form of $t$. 
For orthogonal TRSs computable, normalising strategies exist,
but then also all of the variants of productivity coincide
(see Section~\ref{sec:productivity:subsec:discussion}).

\begin{theorem}\label{thm:weak}
  The recognition problem for weak productivity is $\csig{1}{1}$-complete.
\end{theorem}
\begin{proof}
  For the proof of $\csig{1}{1}$-hardness, let $\tm$ be a Turing machine.
  We exchange the rule
  $\tfunap{\msf{run}}{\tmT}{\picknok{x}}{\picknok{y}} \to \tfunap{\msf{run}}{\bfunap{\tm}{x}{y}}{\picknok{y}}{\pickn}$
  in the TRS $\roottrs$ from~\cite{endr:geuv:zant:2009}
  by 
  the rule $\tfunap{\msf{run}}{\tmT}{\picknok{x}}{\picknok{y}} \to \strcns{0}{\tfunap{\msf{run}}{\bfunap{\tm}{x}{y}}{\picknok{y}}{\pickn}}$.
  Then we obtain that the term $\tfunap{\msf{run}}{\tmT}{\pickn}{\pickn}$ is weakly productive
  if and only if ${>_\tm}$ is not well-founded (which is $\csig{1}{1}$-hard to decide).
  This concludes $\csig{1}{1}$-hardness.

  The remainder of the proof proceeds analogously
  to the proof of Theorem~\ref{thm:omfair},
  except that we now have an existential function quantifier $\exists r$ 
  to quantify over all rewrite sequences of length $\omega$.
  Hence weak productivity is in $\csig{1}{1}$.
  \qed
 \end{proof}

\subsection{Discussion}%
  \label{sec:productivity:subsec:discussion}

For orthogonal TRSs all of the variants of productivity coincide.
That is, if we restrict the first variant to computable outermost-fair strategies;
as already discussed, other strategies are not very reasonable.
For orthogonal TRSs there always exist computable outermost-fair strategies,
and whenever for a term there exists a constructor normal form, 
then it is unique and all outermost-fair rewrite sequences will end in this unique constructor normal form.

This raises the question whether uniqueness of the constructor normal forms
should be part of the definition of productivity.
We consider a specification of the stream of random bits:
\begin{align*}
  \random &\to \strcns{0}{\random} &
  \random &\to \strcns{1}{\random}
\end{align*}
Every rewrite sequence starting from $\random$ ends in a normal form. 
However, these normal forms are not unique. In fact, there are uncountably many of them.
We did not include uniqueness of normal forms
into the definition of productivity since
non-uniqueness only arises in non-orthogonal TRSs 
when using non-deterministic strategies.
However, one might want to require uniqueness of normal forms even in the case of non-orthogonal TRSs.

\begin{theorem}\label{thm:unique}
  The problem of determining, for TRSs~$\atrs$ and terms $t$ in $\atrs$, 
  whether $t$ has a unique (finite or infinite) normal form 
  is $\cpi{1}{1}$-complete.
\end{theorem}

\begin{proof}
  For $\cpi{1}{1}$-hardness, we extend the TRS constructed in the proof of Theorem~\ref{thm:weak} by the rules:
  $\msf{start} \to \tfunap{\msf{run}}{\tmT}{\pickn}{\pickn}$,
  $\tfunap{\msf{run}}{x}{y}{z} \to \tfunap{\msf{run}}{x}{y}{z}$,
  $\msf{start} \to \strff{ones}$, and
  $\strff{ones} \to \strcns{1}{\strff{ones}}$.
  Then $\msf{start}$ has a unique normal form 
  if and only if $>_\tm$ is well-founded.
  For $\cpi{1}{1}$-completeness, we observe that
  the property can be characterised by a $\cpi{1}{1}$-formula:
  we quantify over two infinite rewrite sequences, 
  and, in case both of them end in a normal form, we compare them.
  Note that consecutive universal quantifiers can be compressed into one. \qed
\end{proof}

Let us consider the impact on computational complexity 
of taking up the condition of uniqueness of normal forms
into the definition of productivity.
Including uniqueness of normal forms without considering the strategy
would increase the complexity of productivity with respect to a family of strategies to $\cpi{1}{1}$.
However, we think that doing so would be contrary to the spirit 
of the notion of productivity.
Uniqueness of normal forms should only be required 
for the normal forms
reachable by the given (non-deterministic) strategy.
But then the complexity of productivity remains unchanged,
$\cpi{0}{2}$-complete.
The complexity of strong productivity
remains unaltered, $\cpi{1}{1}$-complete,
when including uniqueness of normal forms.
However, the degree of undecidability of weak productivity increases.
From the proofs of Theorems~\ref{thm:weak} and~\ref{thm:unique}
it follows that the property would then both be
$\csig{1}{1}$-hard and $\cpi{1}{1}$-hard,
then being in $\cdel{1}{1}$.

\section{Productivity for Lazy Stream Specifications is $\cpi{0}{2}$}\label{sec:progs2specs}

In this section we strengthen the undecidability result of Theorem~\ref{thm:strategy}
by showing that the productivity problem is $\cpi{0}{2}$-complete
already for a very simple format of stream specifications,
namely the lazy stream format (LSF) introduced on page~\pageref{page:LSF}.
We do so by giving a translation from Fractran programs into LSF
and applying Theorem~\ref{fractran:pi02}.

\begin{definition}\normalfont
  Let $\aprg = \prglist{\iafrac{1},\ldots,\iafrac{k}}$ be a Fractran program.
  Let $d$ be the least common multiple of the denominators of $\aprg$, 
  that is, $d \defdby  \lcm{\iaden{1},\ldots,\iaden{k}}$.
  Then for $n = 1,\ldots,d$ define
  $\anum_n' = \ianum{i} \cdot (d / \iaden{i})$ and $\aoff_n = n\cdot\iafrac{i}$
  where $\iafrac{i}$ is the first fraction of $\aprg$ such that $n\cdot\iafrac{i}$ is an integer,
  and we let $\anum_n'$ and $\aoff_n$ be undefined if no such fraction exists.
  Then, the \emph{stream specification induced by $\aprg$}
  is a term rewriting system $\mcl{R}_\aprg = \pair{\iasig{\aprg}}{\iarul{\aprg}}$
  with:
  \[
   \iasig{\aprg} 
   = 
   \{ \pebble , \,\sstrcns\, , \shead, \stail , \szipn{d} , \transl{\aprg} \} 
   \cup
   \{ \smodn{\anum_n'} \where \text{$\anum_n'$ is defined} \}
  \]
  and with $\iarul{\aprg}$ consisting of the following rules:
  \begin{gather*}
    \transl{\aprg} \to \zipn{d}{\translarg{1},\ldots,\translarg{d}}, 
    \text{ where, for $1 \le n \le d$, $\translarg{n}$ is shorthand for:} \\
    \translarg{n} =
    \begin{cases}
      \modn{\anum_n'}{\tailn{\aoff_n-1}{\transl{\aprg}}} 
      &\text{if $\anum_n'$ is defined,}\\
      \strcns{\pebble}{\modn{d}{\tailn{n-1}{\transl{\aprg}}}} 
      &\text{if $\anum_n'$ is undefined.}
    \end{cases}
    \\
    \begin{aligned}
    \head{\strcns{x}{\astr}} &\to x 
    &
    \modn{k}{\astr} &\to \strcns{\head{\astr}}{\modn{k}{\tailn{k}{\astr}}} 
    \\
    \tail{\strcns{x}{\astr}} &\to \astr
    &\quad
    \zipn{d}{\iastr{1},\iastr{2}\ldots,\iastr{d}}
    &\to \strcns{\head{\iastr{1}}}{\zipn{d}{\iastr{2},\ldots,\iastr{d},\tail{\iastr{1}}}}
    \end{aligned}
  \end{gather*}
  where $x$, $\astr$, $\iastr{i}$ are variables.\footnote[1]{%
  Note that $\modn{d}{\tailn{n-1}{\zipn{d}{\translarg{1},\ldots,\translarg{d}}}}$ equals $T_n$,
  and so, in case $\anum_n'$ is undefined, we just have 
  $\translarg{n} = \strcns{\pebble}{\translarg{n}}$.
  In order to have the simplest TRS possible (for the purpose at hand),
  we did not want to use an extra symbol $\strff{(\pebble)}$ 
  and rule $\strff{(\pebble)} \to \strcns{\pebble}{\strff{(\pebble)}}$.}
\end{definition}
The rule for $\smodn{n}$ defines a stream function 
which takes from a given stream $\astr$
all elements $\funap{\astr}{i}$ with $\congrmod{i}{0}{n}$, 
and results in a stream consisting of those elements in the original order.
As we only need rules $\smodn{\anum_n'}$ whenever $\anum_n'$ is defined
we need $d$ such rules at most.

\newcommand{\mo}[1]{\funap{\varphi}{#1}}
If $\anum_{n}'$ is undefined then it should be understood that $m \cdot \anum_{n}'$ is undefined.
For $n \in \nat$ let $\mo{n}$ denote the number from $\{1,\ldots,d\}$ with $\congrmod{n}{\mo{n}}{d}$.
\begin{lemma}\label{lem:pn}
  For every $n > 0$ we have 
  $\funap{f_\aprg}{n} = \floor{(n-1)/d} \cdot \anum_{\mo{n}}' + \aoff_{\mo{n}}$.
\end{lemma}
\begin{proof}
  Let $n > 0$.
  For every $i \in \{1,\ldots,k\}$ we have 
  $n \cdot \iafrac{i} \in \nat$
  if and only if
  $\mo{n} \cdot \iafrac{i} \in \nat$, since $n \equiv \mo{n} \mod d$
  and $d$ is a multiple of $\iaden{i}$.
  Assume that $\funap{f_\aprg}{n}$ is defined. Then
  $\funap{f_\aprg}{n} = n \cdot \anum_{\mo{n}}'/d 
   = (\floor{(n-1)/d} \cdot d + ((n-1) \mod d) +1)\cdot \anum_{\mo{n}}'/d$
  $= \floor{(n-1)/d} \cdot \anum_{\mo{n}}' + \mo{n} \cdot \ianum{i} / \iaden{i}
   = \floor{(n-1)/d} \cdot \anum_{\mo{n}}' + \aoff_{\mo{n}}$.
  Otherwise whenever $\funap{f_\aprg}{n}$ is undefined then $\anum_{\mo{n}}'$ is undefined.
  \qed
\end{proof}

\begin{lemma}\label{lem:red:fractran:specs}
  Let $\aprg$ be a Fractran program.
  Then $\mcl{R}_\aprg$ is productive for $\transl{\aprg}$
  if and only if $\aprg$ is terminating on all integers $n > 0$.
\end{lemma}
\begin{proof}
  \newcommand{\tl}[2]{\funap{#1}{#2}}
  Let $\tl{\astr}{n}$ be shorthand for $\head{\tailn{n}{\astr}}$.
  It suffices to show for all $n \in \nat$:
  $\tl{\transl{\aprg}}{n} \to^* \bullet$ 
  if and only if
  $\aprg$ halts on $n$.
  For this purpose we show
  $\tl{\transl{\aprg}}{n} \to^+ \bullet$
  whenever $\funap{f_\aprg}{n+1}$ is undefined,
  and $\tl{\transl{\aprg}}{n} \to^+ \tl{\transl{\aprg}}{\funap{f_\aprg}{n+1}-1}$, otherwise.
  We have $\tl{\transl{\aprg}}{n} \to^* \tl{\translarg{\mo{n+1}}}{\floor{n/d}}$.

  Assume that $\funap{f_\aprg}{n+1}$ is undefined.
  By Lemma~\ref{lem:pn} $\anum_{\mo{n+1}}'$ is undefined, thus
  thus $\tl{\transl{\aprg}}{n} \to^* \bullet$ whenever $\floor{n/d} = 0$, 
  and otherwise we have:
  \begin{align*}
    \tl{\transl{\aprg}}{n} \to^* \tl{\translarg{\mo{n+1}}}{\floor{n/d}} 
      &\to^* \tl{\modn{d}{\tailn{\mo{n+1}-1}{\transl{\aprg}}}}{\floor{n/d}-1}
      \to^* \tl{\transl{\aprg}}{n'}
  \end{align*}
  where $n' = (\floor{n/d}-1)\cdot d + \mo{n+1}-1 = n-d$.
  Clearly $\congrmod{n}{n'}{d}$, and then 
  $\tl{\transl{\aprg}}{n} \to^* \bullet$ follows by induction on $n$.

  Assume that $\funap{f_\aprg}{n+1}$ is defined.
  By Lemma~\ref{lem:pn} $\anum_{\mo{n+1}}'$ is defined and:
  \begin{align*}
    \tl{\transl{\aprg}}{n} &\to^* \tl{\translarg{\mo{n+1}}}{\floor{n/d}}
      \to^* \tl{\modn{\anum_{\mo{n+1}}'}{\tailn{\aoff_{\mo{n+1}}-1}{\transl{\aprg}}}}{\floor{n/d}}
  \end{align*}
  and hence $\tl{\transl{\aprg}}{n} \to^+ \tl{\transl{\aprg}}{n'}$ 
  with $n' = \floor{n/d} \cdot \anum_{\mo{n+1}}' + \aoff_{\mo{n+1}}-1$.
  Then we have $n' = \funap{f_\aprg}{n+1} - 1$ by Lemma~\ref{lem:pn}.
  \qed
\end{proof}

\begin{theorem}\label{thm:prod:specs}
  The restriction of the productivity problem to stream specifications
  induced by Fractran programs and outermost-fair strategies 
  is $\cpi{0}{2}$-complete. 
\end{theorem}

\begin{proof}
  Since by Lemma~\ref{lem:red:fractran:specs} the uniform halting problem
  for Fractran programs can be reduced to the problem here, 
  $\cpi{0}{2}$-hardness is a consequence of \mbox{Theorem~\ref{fractran:pi02}}.
  $\cpi{0}{2}$-completeness follows from membership of the problem
  in $\cpi{0}{2}$, which can be established analogously as in the proof 
  of Theorem~\ref{thm:strategy}.
  \qed
\end{proof}
Note that Theorem~\ref{thm:prod:specs} also gives rise to
an alternative proof for the $\cpi{0}{2}$-hardness part 
of Theorem~\ref{thm:strategy}, the result concerning the 
computational complexity of productivity with respect to strategies.

\bibliography{main}

\end{document}